\documentclass[reqno,12pt]{amsart}
\usepackage{amssymb}
\usepackage{amsmath}
\usepackage{amsthm}
\usepackage{mathtools}
\usepackage{amsfonts}
\usepackage{xcolor}
\usepackage{graphicx}
\usepackage{tikz}
\usepackage{comment}
\usepackage{microtype}
\usepackage{enumitem}
\usepackage[
    colorlinks=true,
    linkcolor={blue},
    citecolor={blue},
    urlcolor={black!30!blue}
]{hyperref}
\usepackage[alphabetic,initials]{amsrefs}
\renewcommand{\MR}[1]{}

\title[Sharp Quantum Dynamics]{Sharp Logarithmic Quantum Dynamics for Quasiperiodic Schr\"odinger Operators}

\author[W.\ Liu]{Wencai Liu}
\address{[W.\ Liu] Department of Mathematics, Texas A\&M University, College Station, TX 77843, USA}
\email{\href{mailto:wencail@tamu.edu}{wencail@tamu.edu}}

\author[X.\ Wang]{Xueyin Wang}
\address{[X.\ Wang] Department of Mathematics, Texas A\&M University, College Station, TX 77843, USA}
\email{\href{mailto:xueyin@tamu.edu}{xueyin@tamu.edu}}

\theoremstyle{plain}
\newtheorem{theorem}{Theorem}[section]

\newtheorem{corollary}[theorem]{Corollary}
\newtheorem{lemma}[theorem]{Lemma}
\newtheorem{proposition}[theorem]{Proposition}

\theoremstyle{definition}
\newtheorem{definition}{Definition}[section]

\newtheorem{remark}[theorem]{Remark}

\numberwithin{equation}{section}

\begin{document}

\begin{abstract}
    Dynamical localization requires all position moments of a quantum wavepacket to remain bounded in time, but for quasiperiodic Schr\"odinger operators such bounds are generally not uniform in phase. In the positive Lyapunov exponent regime, the best known phase-uniform estimates instead grow on a logarithmic scale. We prove that both the logarithmic scale and the dependence on the moment order are sharp for a class of one-frequency quasiperiodic Schr\"odinger operators with even potentials. Our main ingredient is a reflective version of semi-uniformly localized eigenfunctions, adapted to the two localization centers forced by a completely resonant phase, from which we obtain matching logarithmic lower bounds along sequences of times.
\end{abstract}

\maketitle	
\section{Introduction}

In this paper, we study the spreading of wavepackets under the Schr\"odinger evolution. We consider the discrete Schr\"odinger operator on $\ell^{2}(\mathbb{Z})$ given by
\begin{equation}\label{H}
    [H\psi](n)=\psi(n+1)+\psi(n-1)+V(n)\psi(n).
\end{equation}
Throughout the introduction, $v$ denotes a real-valued sampling function on $\mathbb{T}$, while $V$ denotes the real-valued potential sequence; in the quasiperiodic setting,
\begin{equation*}
    V(n)=v(\theta+n\alpha).
\end{equation*}
For $p>0$, the spreading of a wavepacket initially localized at the origin is measured by the $p$-th moment of the position operator,
\begin{equation}\label{moment}
    \langle |X_{H}|^{p}\rangle (t)=\sum_{n\in\mathbb{Z}}|n|^{p}|\langle\delta_{n},e^{-itH}\delta_{0}\rangle|^{2},
\end{equation}
and its Abel average version,
\begin{equation}\label{avmoment}
    \langle |\widetilde{X}_{H}|^{p}\rangle(T)=\frac{2}{T}\int_{0}^{\infty}e^{-2t/T}  \langle |X_{H}|^{p}\rangle (t) \,\mathrm{d}t.
\end{equation}
The long-time behavior of \eqref{moment} and \eqref{avmoment} reflects the propagation speed of the diffusion \cite{MR4070305}. For further details, we refer the readers to the monographs of Damanik and Fillman \cite{MR4567742,MR4840232}.

\subsection{From SULE to reflective SULE}
Anderson localization describes the insulating behavior of a particle in a disordered system, characterized by the pure point spectrum with exponentially decaying eigenfunctions. It has been established for broad classes of Schr\"odinger operators in \cite{MR1740982,MR1815703,MR4181827,MR3707287,MR4108613,MR3779957}. More precisely, Anderson localization provides a complete orthonormal basis of eigenfunctions $\{\phi_{s}\}_{s\in\mathbb{Z}}$ such that, for every $s\in\mathbb{Z}$, there exist a localization center $m_{s}\in\mathbb{Z}$ and constants $C_{s}>0$ and $\gamma_{s}>0$ satisfying
\begin{equation*}
    |\phi_{s}(n)|\leqslant C_{s}e^{-\gamma_{s}|n-m_{s}|}
\end{equation*}
for every $n\in\mathbb{Z}$. This estimate controls each eigenfunction separately.

Dynamical localization is a  property concerning the full Schr\"odinger evolution. It asserts that an initially localized wavepacket remains localized uniformly in time. The operator $H$ exhibits dynamical localization if for every $p>0$,
\begin{equation*}
    \sup_{t\in\mathbb{R}}\langle |X_{H}|^{p}\rangle (t)<\infty.
\end{equation*}
Various forms of dynamical localization for Schr\"odinger operators have been established in \cite{MR4637128,MR2100420,MR4216568}.

Anderson localization does not imply dynamical localization in general \cite{MR1428099}. Indeed, the spectral expansion of $e^{-itH}$ requires uniform control of the entire eigenbasis. To formulate such control, del Rio, Jitomirskaya, Last, and Simon introduced the notion of semi-uniformly localized eigenfunctions (SULE) in \cite{MR1428099}.

\begin{definition}[SULE]
Let $H$ be a self-adjoint operator on $\ell^{2}(\mathbb{Z})$ with a complete orthonormal basis of eigenfunctions $\{\phi_{s}\}_{s\in\mathbb{Z}}$. We say that $H$ has SULE if there exists $\gamma>0$ such that, for every $\varepsilon>0$, there exists $C_{\varepsilon}>0$ satisfying
\begin{equation*}
    |\phi_{s}(n)|\leqslant C_{\varepsilon}e^{\varepsilon|m_{s}|}e^{-\gamma|n-m_{s}|}.
\end{equation*}
\end{definition}

The factor $e^{\varepsilon|m_{s}|}$ allows a subexponential growth with respect to the localization center, while the decay rate $\gamma$ is uniform over the entire eigenbasis.  It has been shown that SULE implies dynamical localization \cite{MR1428099}.

The one-center structure in the definition of SULE is essential: it describes an eigenfunction with one peak and exponential decay away from that peak, as illustrated in Figure~\ref{fig:localization-centers}(a). Jitomirskaya, Liu, and Mi \cite{jitomirskaya2024sharp} showed that phase resonances for quasiperiodic operators with even sampling functions can rule out SULE. To see geometrically why such a resonance obstructs SULE, first consider a potential   satisfying
\begin{equation*}
    V(n)=V(1-n),\qquad n\in\mathbb{Z}.
\end{equation*}
Then the corresponding Schr\"odinger operator is invariant under the reflection $n\mapsto1-n$. Since one-dimensional discrete Schr\"odinger operators have simple eigenvalues, every eigenfunction $\psi$ satisfies
\begin{equation*}
    |\psi(n)|=|\psi(1-n)|,\qquad n\in\mathbb{Z}.
\end{equation*}
Consequently, if $m_{s}$ is a global maximum of $|\phi_{s}|$, then $1-m_{s}$ is also a global maximum. Thus the eigenfunction has two reflection-related localization centers. Across the eigenbasis, the distance between these centers need not remain bounded, so a uniform one-center SULE estimate cannot capture this geometry. Figure~\ref{fig:localization-centers}(b) illustrates the resulting two-center profile.

This reflection symmetry arises naturally at completely resonant phases. Let $v\in C(\mathbb{T},\mathbb{R})$ be even and consider the quasiperiodic operator
\begin{equation*}
    [H_{v,\alpha,\theta}\psi](n) = \psi(n+1) + \psi(n-1) + v(\theta+n\alpha)\psi(n).
\end{equation*}
A phase is called completely resonant when $2\theta\in\alpha\mathbb{Z}+\mathbb{Z}$. In this paper, our starting operator corresponds to
\begin{equation}\label{halfalpha}
    \theta_{0}=-\frac{\alpha}{2}.
\end{equation}
Writing $V(n)=v(\theta_{0}+n\alpha)$, the evenness of $v$ gives
\begin{equation*}
    V(n)=V(1-n),\qquad n\in\mathbb{Z}.
\end{equation*}
Thus the completely resonant operator realizes precisely the two-center geometry described above. This motivates the following reflective version of SULE.

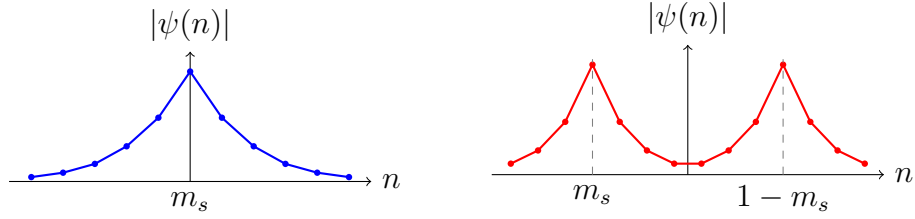
\begin{figure}[t]
\centering
\begin{minipage}[t]{0.47\textwidth}
\centering
\begin{tikzpicture}[x=0.42cm,y=1.45cm]
    \draw[->] (-5.7,0) -- (5.7,0) node[right] {$n$};
    \draw[->] (0,0) -- (0,1.18) node[above] {$|\psi(n)|$};
    \draw[blue,thick] plot coordinates {(-5,0.04) (-4,0.08) (-3,0.16) (-2,0.32) (-1,0.58) (0,1) (1,0.58) (2,0.32) (3,0.16) (4,0.08) (5,0.04)};
    \foreach \x/\y in {-5/0.04,-4/0.08,-3/0.16,-2/0.32,-1/0.58,0/1,1/0.58,2/0.32,3/0.16,4/0.08,5/0.04}
        \fill[blue] (\x,\y) circle (1.2pt);
    \node[below] at (0,0) {$m_s$};
\end{tikzpicture}

\small (a) SULE: one localization center.
\end{minipage}
\hfill
\begin{minipage}[t]{0.49\textwidth}
\centering
\begin{tikzpicture}[x=0.36cm,y=1.45cm]
    \draw[->] (-6.7,0) -- (7.7,0) node[right] {$n$};
    \draw[->] (0.5,0) -- (0.5,1.18) node[above] {$|\psi(n)|$};
    \draw[red,thick] plot coordinates {(-6,0.10) (-5,0.22) (-4,0.48) (-3,1) (-2,0.48) (-1,0.22) (0,0.10) (1,0.10) (2,0.22) (3,0.48) (4,1) (5,0.48) (6,0.22) (7,0.10)};
    \foreach \x/\y in {-6/0.10,-5/0.22,-4/0.48,-3/1,-2/0.48,-1/0.22,0/0.10,1/0.10,2/0.22,3/0.48,4/1,5/0.48,6/0.22,7/0.10}
        \fill[red] (\x,\y) circle (1.2pt);
    \draw[dashed,gray] (-3,0) -- (-3,1.05);
    \draw[dashed,gray] (4,0) -- (4,1.05);
    \node[below] at (-3,0) {$m_s$};
    \node[below] at (4,0) {$1-m_s$};
\end{tikzpicture}

\small (b) Reflective SULE: two localization centers.
\end{minipage}
\caption{One center vs two centers: exact reflection symmetry forces the two centers in (b) to have the same height.}
\label{fig:localization-centers}
\end{figure}
\begin{definition}[Reflective SULE]\label{reflectivelocalization}
Let $v\in C(\mathbb{T},\mathbb{R})$ be even and let $\alpha\in\mathbb{R}\setminus\mathbb{Q}$. Set $H_{0}\coloneqq H_{v,\alpha,-\alpha/2}$. We say that $H_{0}$ has reflective SULE if
\begin{enumerate}
    \item The operator $H_{0}$ has a complete orthonormal basis of real eigenfunctions $\{\phi_{s}\}_{s\in\mathbb{Z}}$ with corresponding eigenvalues $\{E_{s}\}_{s\in\mathbb{Z}}$.
    \item For every $s\in\mathbb{Z}$, one can choose a global maximum point $m_{s}\leqslant0$ of $|\phi_{s}|$. There exist $0<\epsilon<\gamma$ and  $C_{\epsilon}>0$ such that 
    \begin{equation}\label{criteriondecay}
        |\phi_{s}(n)|\leqslant C_{\epsilon}e^{\epsilon|m_{s}|}e^{-\gamma d_{s}(n)},
    \end{equation}
    where
    \begin{equation}\label{defds}
        d_{s}(n)\coloneqq\min\big\{|n-m_{s}|,|n-(1-m_{s})|\big\}.
    \end{equation}
\end{enumerate}
\end{definition}
The quantity $d_{s}(n)$ defined in \eqref{defds} represents the distance from $n$ to the reflection pair $\{m_{s},1-m_{s}\}$. Unlike SULE, reflective SULE is more flexible and does not require \eqref{criteriondecay} to hold for every $\epsilon>0$, but only for some fixed constants $0<\epsilon<\gamma$ uniformly over the entire eigenbasis.

\subsection{Sharp quantum dynamics}
Two widely studied forms of localization are Anderson localization and dynamical localization. The latter requires the moments in \eqref{moment} to remain bounded in time. For quasiperiodic operators, many localization results hold for almost every phase $\theta$, but such phase-by-phase conclusions do not give a bound uniform in $\theta$; see, for example, \cite{MR1740982,MR1815703,MR2100420,MR4181827,MR4216568,MR4637128}. If one requires uniformity in phase, the known bounds in the positive Lyapunov exponent regime are instead logarithmic. This leads to a natural question: is the logarithmic scale sharp, or could it be improved, for instance, to a power of $\ln\ln T$? It is also natural to ask whether the dependence on the moment order $p$ is sharp. We show that, for a family of even quasiperiodic potentials, both the logarithmic scale and the exponent $p$ are sharp.

A general approach to lower bounds, introduced by Last \cite{MR1423040}, relates quantum dynamics to the Hausdorff dimensional properties of spectral measures. In the zero-dimensional regime, Damanik and Landrigan \cite{MR1963769} established a criterion for continuity with respect to logarithmic Hausdorff measures and applied it to Sturmian operators, obtaining logarithmic quantum dynamical lower bounds; see also \cite{MR4578350} for a generalization. More recently, Jitomirskaya and Zhang \cite{MR4404788} introduced a quantitative almost periodicity condition and proved lower spectral dimension and quantum dynamical bounds for potentials with sufficiently strong repetitions.

These spectral dimensional arguments use subordinacy theory to relate dynamics to eigenfunction norm growth. Our reflective SULE provides further structural information about the localization centers. Consequently, the double-center mechanism yields a lower bound on the propagation speed of wavepackets.

For convenience, we write $\langle |X_{\theta}|^{p}\rangle$ for the $p$-th moment associated with $H_{v,\alpha,\theta}$ and $\langle |\widetilde{X}_{\theta}|^{p}\rangle$ for its Abel average.  We first establish a general criterion for logarithmic lower bounds based on reflective SULE. 

\begin{theorem}\label{criteriontheorem}
Let $v\in C(\mathbb{T},\mathbb{R})$ be even and $\alpha\in\mathbb{R}\setminus\mathbb{Q}$. Assume that $H_{v,\alpha,-\alpha/2}$ has reflective SULE. Then there exists a dense set of phases $\Theta\subseteq\mathbb{T}$ such that, for every $\theta\in\Theta$, there are sequences $T_{j}\to\infty$, $t_{j}\to\infty$ for which, for every $p>0$,
\begin{equation}\label{criterionmoment}
    \begin{split}
        \langle |X_{\theta}|^{p}\rangle(t_{j})&\geqslant c_{p}(\ln t_{j})^{p},\\
        \langle |\widetilde{X}_{\theta}|^{p}\rangle(T_{j})&\geqslant c_{p}(\ln T_{j})^{p}.
    \end{split}
\end{equation}
\end{theorem}

Establishing upper bounds has attracted substantial attention. Various approaches, including transfer matrix methods, large deviation estimates, semi-algebraic techniques, and finite-volume Green's function estimates, have yielded $T^{\varepsilon}$ and $(\ln T)^{C}$ bounds for broad classes of ergodic and quasiperiodic operators \cite{MR3869380,MR4564259,MR4604835,MR4944474,MR5026195,MR5029966,MR4288185,MR2291919,MR5000719,Liouville}. These results include long-range operators, higher-dimensional lattice models, multi-frequency shifts, and skew-shifts. Recall that a frequency $\alpha$ is said to be Diophantine, denoted by $\alpha\in\mathrm{DC}$, if there exist constants $\gamma>0$ and $\tau>1$ such that
\begin{equation*}
    \|k\alpha\|_{\mathbb{T}}\geqslant \frac{\gamma }{|k|^{\tau}}.
\end{equation*}
for every $k\neq 0$. In the positive Lyapunov exponent regime, results from \cite{MR4604835,MR4546503,MR5029966} yield the following upper bounds, uniformly in the phase, for $\alpha\in\mathrm{DC}$ and $v\in C^{\omega}(\mathbb{T},\mathbb{R})$:
\begin{equation}\label{upper}
    \begin{split}
        \sup_{\theta\in\mathbb{T}} \langle |X_{\theta}|^{p}\rangle(t)&\leqslant C_{p,\varepsilon}(\ln t)^{p+\varepsilon},\\
        \sup_{\theta\in\mathbb{T}} \langle |\widetilde{X}_{\theta}|^{p}\rangle(T)&\leqslant C_{p,\varepsilon}(\ln T)^{p+\varepsilon}.
    \end{split}
\end{equation}
Combining these upper bounds with \hyperref[criteriontheorem]{Theorem~\ref*{criteriontheorem}} shows that both the logarithmic scale and the exponent $p$ are sharp.

\begin{theorem}\label{sharpthm}
    Let $\alpha\in \mathrm{DC}$ and $v\in C^{\omega}(\mathbb{T},\mathbb{R})$ be even. Assume the Lyapunov exponent is positive on the spectrum and $H_{v,\alpha,-\alpha/2}$ has reflective SULE. Then for every $p>0$,
\begin{equation}\label{sharpexponent}
    \limsup_{t\to\infty}\frac{\ln \sup_{\theta}\langle |X_{\theta}|^{p}\rangle(t)}{p\ln \ln t} = \limsup_{T\to\infty}\frac{\ln\sup_{\theta}\langle |\widetilde{X}_{\theta}|^{p}\rangle(T)}{p\ln \ln T} = 1.
\end{equation}
\end{theorem}

As an application, we obtain the sharp quantum dynamics for the supercritical almost Mathieu operator
\begin{equation}\label{amo}
    [H_{\lambda,\alpha,\theta}\psi](n)=\psi(n+1)+\psi(n-1)+2\lambda\cos2\pi(\theta+n\alpha)\psi(n),
\end{equation}
where $\lambda>1$ is the coupling constant.  

\begin{theorem}\label{amothm}
    Let $\alpha\in\mathrm{DC}$ and $\lambda>1$. Then there exists a dense set of phases $\Theta\subseteq\mathbb{T}$ such that, for every $\theta\in\Theta$, there are sequences $T_{j}\to\infty$, $t_{j}\to\infty$ for which \eqref{criterionmoment} holds. Moreover,  for every $p>0$,
\begin{equation}\label{sharpexponentliu}
    \limsup_{t\to\infty}\frac{\ln \sup_{\theta}\langle |X_{\theta}|^{p}\rangle(t)}{p\ln \ln t} = \limsup_{T\to\infty}\frac{\ln\sup_{\theta}\langle |\widetilde{X}_{\theta}|^{p}\rangle(T)}{p\ln \ln T} = 1.
\end{equation}
\end{theorem}

\subsection{Idea of the proof}
Now we explain the idea of the proof. Let $\theta_{0}$ be defined in \eqref{halfalpha}, and set
\begin{equation*}
    \theta_{k}\coloneqq\theta_{0}-k\alpha,\qquad H_{k}\coloneqq H_{v,\alpha,\theta_{k}},\qquad H_{0}\coloneqq H_{v,\alpha,\theta_{0}}.
\end{equation*}
The proof consists of three steps. We first establish the required estimate for the dynamics of $H_{0}$ with the initial state $\delta_{-k}$. Because every eigenfunction of $H_{0}$ is reflectively symmetric about $1/2$, the corresponding reflection exchanges the regions $n\leqslant 0$ and $n\geqslant 1$. Thus, every eigenfunction has exactly half of its mass in the region $n\geqslant 1$. By counting localization centers and estimating the eigenfunction tails, we identify a family of eigenfunctions whose mass at $-k$ is sufficiently large. Applying the spectral expansion, we show in \hyperref[half]{Proposition~\ref*{half}} that these eigenfunctions yield the required lower bound for the mass of $e^{-itH_{0}}\delta_{-k}$ in the region $n\geqslant 1$.

We then transfer this estimate to $H_{k}$ with the initial state $\delta_{0}$. Indeed, \hyperref[Qk]{Corollary~\ref*{Qk}} shows that the translation covariance transfers this lower bound to the mass of $e^{-itH_{k}}\delta_{0}$ in the region $n\geqslant k+1$.

Finally, we construct a single phase $\theta$ as the limit of a rapidly convergent subsequence of the phases $\theta_{k}$, see \hyperref[phase]{Lemma~\ref*{phase}}. The convergence rate is chosen to be sufficiently fast so that the dynamical estimates at these phases remain valid at $\theta$ on the corresponding time scales. Since the required time scale grows exponentially with the propagation distance $k$, this directly yields the logarithmic lower bounds for the quantum dynamics.

\section{Estimate of eigenfunctions}
From now on, we assume the hypotheses of \hyperref[criteriontheorem]{Theorem~\ref*{criteriontheorem}}.
For brevity, we write
\begin{equation*}
    H_{0}\coloneqq H_{v,\alpha,\theta_{0}},\qquad \theta_{0}=-\frac{\alpha}{2}.
\end{equation*}
Fix $0<\epsilon<\gamma$ and $C_{\epsilon}>0$ such that \eqref{criteriondecay} holds.

\subsection{Reflective symmetry}

\begin{lemma}\label{eigenbasis}
    Let $\{(\phi_{s},E_{s})\}_{s\in\mathbb{Z}}$ be the eigen-pairs in \hyperref[reflectivelocalization]{Definition~\ref*{reflectivelocalization}}.  Then for each $s\in\mathbb{Z}$, there exists  $\sigma_{s}\in \{\pm 1\}$ such that
    \begin{equation*}
        \phi_{s}(1-n)=\sigma_{s}\phi_{s}(n),\quad \text{for all} \quad n\in\mathbb{Z}.
    \end{equation*}
\end{lemma}

\begin{proof}

    By our choice of $\theta_{0}$, the operator $H_{0}$ exhibits reflection symmetry with respect to the transformation $n \mapsto 1-n$. More precisely, let $\widetilde{\phi}_{s}(n)=\phi_{s}(1-n)$ for every $s\in\mathbb{Z}$. 
    Since $v$ is even and
    
    \begin{equation*}
        \theta_{0}+(1-n)\alpha=-(\theta_{0}+n\alpha),
    \end{equation*}
    we have
    \begin{equation*}
        v(\theta_{0}+(1-n)\alpha)=v(\theta_{0}+n\alpha).
    \end{equation*}
    Combine this with $H_{0}\phi_{s}=E_{s}\phi_{s}$,
    a direct computation verifies that
    \begin{equation*}
        H_{0}\widetilde{\phi}_{s}=E_{s}\widetilde{\phi}_{s}\quad \text{for every}\quad s\in\mathbb{Z}.
    \end{equation*}
    Because the eigenvalues of one-dimensional discrete Schr\"odinger operators are simple, the eigenspace corresponding to $E_{s}$ is one-dimensional. Therefore, there exists a constant $\sigma_{s}\in\mathbb{R}$ such that $\widetilde{\phi}_{s}=\sigma_{s}\phi_{s}$. Evaluating this symmetry twice for any $n\in\mathbb{Z}$ yields
    \begin{equation*}
        \phi_{s}(1-n)=\sigma_{s}\phi_{s}(n)=\sigma_{s}^{2}\phi_{s}(1-n).
    \end{equation*}
    This implies $\sigma_{s}^{2}=1$, meaning $\sigma_{s}\in \{\pm 1\}$.
\end{proof}

\hyperref[eigenbasis]{Lemma~\ref*{eigenbasis}} implies that $|\phi_{s}|$ has at least two the global maximum points $m_{s}$ and $1-m_{s}$. Without loss the generality, we assume that $m_{s}\leqslant 0$ and call $m_{s}$ the localization center.

\subsection{Counting centers}
We have the following linear bound on the number of eigenfunction whose localization centers lie in a prescribed interval. The idea of the proof comes from \cite[Theorem 7.1]{MR1428099}.

\begin{lemma}\label{count}
    There exist constant $R_{0}>0$ such that
\begin{equation}\label{countingestimate}
    \#\{s:-R\leqslant m_{s}\leqslant 0\}\leqslant 10 R,
\end{equation}
for every $R\geqslant R_{0}$.
\end{lemma}
\begin{proof}
    Recall that $0<\epsilon<\gamma$. Choose $1<D<2$ such that
    \begin{equation}\label{Dxi}
        \gamma(D-1)>\epsilon.
    \end{equation}
    For $R>0$, we define the interval
    \begin{equation*}
        I_{R}=[-DR,DR]\cap \mathbb{Z}.
    \end{equation*}
    If $-R\leqslant m_{s} \leqslant 0$ and $n\notin I_{R}$, then by \eqref{defds}
    \begin{equation*}
        d_{s}(n)\geqslant (D-1)R-2.
    \end{equation*}
    Thus \eqref{criteriondecay} gives
    \begin{equation*}
        \begin{split}
            \sum_{n\notin I_{R}} |\phi_{s}(n)|^{2}&\leqslant C_{\epsilon}^{2} e^{2\epsilon R} \sum_{n\notin I_{R}}e^{-2\gamma d_{s}(n)}\\
            &\leqslant 4C_{\epsilon}^{2}e^{2\epsilon R} \sum_{\ell\geqslant (D-1)R-2} e^{-2\gamma \ell}\\
            &\leqslant C_{0}e^{-2(\gamma(D-1)-\epsilon)R}.
        \end{split}
    \end{equation*}
    By \eqref{Dxi}, for sufficiently large $R$, we have
    \begin{equation*}
        \sum_{n\notin I_{R}} |\phi_{s}(n)|^{2}\leqslant \frac{1}{2}\quad \text{if}\ -R\leqslant m_{s}\leqslant 0.
    \end{equation*}
    Hence
    \begin{equation*}
        \sum_{n\in I_{R}} |\phi_{s}(n)|^{2}\geqslant \frac{1}{2}\quad \text{if}\ -R\leqslant m_{s}\leqslant 0.
    \end{equation*}
    Using orthonormality and completeness of $\{\phi_{s}\}_{s\in\mathbb{Z}}$, we obtain
\begin{equation*}
    \begin{split}
        \frac{1}{2}\#\{s:-R\leqslant m_{s}\leqslant 0\}&\leqslant\sum_{s:-R\leqslant m_{s}\leqslant 0}\sum_{n\in I_{R}}|\phi_{s}(n)|^{2}\\
        &\leqslant\sum_{n\in I_{R}}\sum_{s\in\mathbb{Z}}|\phi_{s}(n)|^{2}=\# I_{R}.
    \end{split}
\end{equation*}
Since $\# I_{R}\leqslant 5 R$, this proves the result.
\end{proof}

\begin{lemma}\label{tail}
    Let $D'>1$ be such that 
    \begin{equation}\label{D'}
        (\gamma-\epsilon)D'>\gamma.
    \end{equation}
    There exist constants $C_{1},c_{1}>0$ such that
    \begin{equation*}
        \sum_{s: m_{s}\leqslant -D'k}|\phi_{s}(-k)|^{2}\leqslant C_{1}e^{-c_{1}k}
    \end{equation*}
    for all sufficiently large $k\in\mathbb{N}$.
\end{lemma}

\begin{proof}
    If $m_{s}<-k$, then $d_{s}(-k)=-k-m_{s}>0$ for $k\in \mathbb{N}$. By \eqref{criteriondecay}, 
    \begin{equation}\label{-k}
        |\phi_{s}(-k)|^{2}\leqslant C_{\epsilon}^{2}e^{-2\epsilon m_{s}} e^{2\gamma (k+m_{s})} =C_{\epsilon}^{2}e^{2\gamma k}e^{2(\gamma-\epsilon)m_{s}}. 
    \end{equation}
    By \hyperref[count]{Lemma~\ref*{count}}, for any sufficiently large $r\geqslant R_{0}$,
    \begin{equation}\label{counteq}
        \# \{s: m_{s}=-r\}\leqslant \# \{s: -r\leqslant m_{s}\leqslant 0\}\leqslant 10r.
    \end{equation}
    Therefore, by \eqref{D'}, \eqref{-k} and \eqref{counteq}, for sufficiently large $k$,
    \begin{equation}
        \sum_{s:  m_{s}\leqslant -D'k} |\phi_{s}(-k)|^{2}\leqslant  C_{\epsilon}^{2}e^{2\gamma k}\sum_{r\geqslant D'k} 10re^{-2(\gamma-\epsilon)r}\leqslant C_{1}e^{-c_{1}k}.
    \end{equation}
\end{proof}

\section{Separation of eigenvalues}

Let $\rho>0$ and define
\begin{equation*}
    S_{k,\rho}=\{s:  -D'k\leqslant m_{s} \leqslant 0,\, |\phi_{s}(-k)|^{2}\geqslant e^{-\rho k} \}.
\end{equation*}
We first provide some necessary estimates of $S_{k,\rho}$ for later use.
\begin{lemma}\label{restS}
    For all sufficiently large $k\in\mathbb{Z}$,
    \begin{equation*}
        \# S_{k,\rho} \leqslant 10 D' k,
    \end{equation*}
    and there exist constants $C_{2},c_{2}>0$ such that
    \begin{equation*}
        \sum_{s\notin S_{k,\rho}} |\phi_{s}(-k)|^{2}\leqslant C_{1}e^{-c_{1} k}+10 D'ke^{-\rho k}\leqslant C_{2}e^{-c_{2}k}.
    \end{equation*}
\end{lemma}
\begin{proof}
    By \hyperref[count]{Lemma~\ref*{count}}, one has
    \begin{equation}\label{appcount}
        \# S_{k,\rho} \leqslant \# \{s:-D' k\leqslant m_{s}\leqslant 0\}\leqslant 10 D' k.
    \end{equation}

    Let
    \begin{equation*}
        \begin{split}
            S_{1}&=\{s:m_{s}\leqslant -D' k\},\\
            S_{2}&=\{s:  -D'k\leqslant m_{s} \leqslant 0,\, |\phi_{s}(-k)|^{2}\leqslant e^{-\rho k} \}.
        \end{split}
    \end{equation*}
Then
    \begin{equation*}
        \sum_{s\notin S_{k,\rho}}|\phi_{s}(-k)|^{2}\leqslant \sum_{s\in S_{1}}|\phi_{s}(-k)|^{2}+\sum_{s\in S_{2}}|\phi_{s}(-k)|^{2}.
    \end{equation*}
    By \hyperref[tail]{Lemma~\ref*{tail}}, we have
    \begin{equation}\label{S1}
        \sum_{s\in S_{1}}|\phi_{s}(-k)|^{2}\leqslant C_{1}e^{-c_{1}k}.
    \end{equation}
    By \eqref{appcount},
    \begin{equation}\label{S2}
        \sum_{s\in S_{2}}|\phi_{s}(-k)|^{2}\leqslant 10 D'k e^{-\rho k}.
    \end{equation}
    Combining \eqref{S1} with \eqref{S2} completes the proof.
\end{proof}

Denote by $\Sigma$ the spectrum of $H_{0}$. Since $\Sigma\subseteq [-2-\|v\|_{C^{0}},2+\|v\|_{C^{0}}]$, we have, for any $E\in\Sigma$ and $\theta\in\mathbb{T}$,
\begin{equation}\label{B}
    \|S_{E}^{v}(\theta)\|\leqslant 10(\|v\|_{C^{0}}+1)\eqqcolon B,
\end{equation}
where
\begin{equation*}
    S_{E}^{v}(\theta)\coloneqq \begin{pmatrix}
        E-v(\theta)&-1\\
        1&0
    \end{pmatrix}.
\end{equation*}

\begin{lemma}\label{gap}
    Let $s,\ell\in S_{k,\rho}$ and assume $\sigma_{s}\neq \sigma_{\ell}$. Then
    \begin{equation*}
        |E_{s}-E_{\ell}|\geqslant 2B^{-2k} e^{-\rho k}.
    \end{equation*}
\end{lemma}

\begin{proof}
    Recall that the Wronskian of $\phi_{s}$ and $\phi_{\ell}$ is given by
    \begin{equation*}
        W(n)=\phi_{s}(n+1)\phi_{\ell}(n)-\phi_{s}(n)\phi_{\ell}(n+1).
    \end{equation*}
    By $H_{0}\phi_{s}=E_{s}\phi_{s}$ and $H_{0}\phi_{\ell}=E_{\ell}\phi_{\ell}$, a direct calculation gives
\begin{equation}\label{wronskiandifference}
    W(n)-W(n-1)=(E_{s}-E_{\ell})\phi_{s}(n)\phi_{\ell}(n).
\end{equation}
Since $\phi_{s},\phi_{\ell}\in\ell^{2}(\mathbb{Z})$, we have $W(n)\to 0$ as $n\to\infty$. Summing \eqref{wronskiandifference} over $n\geqslant1$, we obtain
\begin{equation}\label{wronskiansum}
    -W(0)=(E_{s}-E_{\ell})\sum_{n\geqslant1}\phi_{s}(n)\phi_{\ell}(n).
\end{equation}
Because $\sigma_{s},\sigma_{\ell}\in \{\pm 1\}$ and $\sigma_{s}\neq\sigma_{\ell}$, \hyperref[eigenbasis]{Lemma~\ref*{eigenbasis}} implies
\begin{equation}\label{W0identity}
    |W(0)|=|\phi_{s}(1)\phi_{\ell}(0)-\phi_{s}(0)\phi_{\ell}(1)|=2|\phi_{s}(0)\phi_{\ell}(0)|.
\end{equation}
Moreover, \hyperref[eigenbasis]{Lemma~\ref*{eigenbasis}}  again gives that
\begin{equation*}
    \sum_{n\geqslant1}|\phi_{s}(n)|^{2}=\sum_{n\geqslant1}|\phi_{\ell}(n)|^{2}=\frac{1}{2}.
\end{equation*}
Thus by H\"older inequality,
\begin{equation}\label{halfcauchy}
    \bigg|\sum_{n\geqslant1}\phi_{s}(n)\phi_{\ell}(n)\bigg|\leqslant \bigg(\sum_{n\geqslant1}|\phi_{s}(n)|^{2}\bigg)^{\frac{1}{2}} \bigg(\sum_{n\geqslant1}|\phi_{\ell}(n)|^{2}\bigg)^{\frac{1}{2}} \leqslant\frac{1}{2}.
\end{equation}
Combine \eqref{wronskiansum},\eqref{W0identity}, with \eqref{halfcauchy}, we get
\begin{equation}\label{Est}
    |E_{s}-E_{\ell}|\geqslant 4|\phi_{s}(0)\phi_{\ell}(0)|.
\end{equation}

In the following, we estimate $|\phi_{s}(0)|$, the estimate of $|\phi_{\ell}(0)|$ is similar. Set
\begin{equation*}
    \Phi_{s}(n)=\begin{pmatrix}
        \phi_{s}(n+1)\\
        \phi_{s}(n)
    \end{pmatrix}.
\end{equation*}
Since $s\in S_{k,\rho}$, we have
\begin{equation*}
    \|\Phi_{s}(-k)\|\geqslant |\phi_{s}(-k)|\geqslant e^{-\rho k/2}.
\end{equation*}
By \eqref{B} and  the cocycle iteration
\begin{equation*}
    \begin{pmatrix}
        \phi_{s}(1)\\ \phi_{s}(0)
    \end{pmatrix}
    =S_{E_{s}}^{v}(\theta_{0})\cdots S_{E_{s}}^{v}(\theta_{0}-(k-1)\alpha) \begin{pmatrix}
        \phi_{s}(-k+1)\\ \phi_{s}(-k)
    \end{pmatrix},
\end{equation*}
we have
\begin{equation*}
    \|\Phi_{s}(0)\| \geqslant B^{-k} \|\Phi_{s}(-k)\| \geqslant B^{-k}e^{-\rho k/2}.
\end{equation*}
\hyperref[eigenbasis]{Lemma~\ref*{eigenbasis}} implies that $\phi_{s}(1)=\sigma_{s}\phi_{s}(0)$, and hence
\begin{equation*}
    \|\Phi_{s}(0)\|=\sqrt{2} |\phi_{s}(0)|.
\end{equation*}
Thus for any $s,\ell\in S_{k,\rho}$, we have
\begin{equation}\label{phi0}
    |\phi_{s}(0)| \geqslant \frac{1}{\sqrt{2}}B^{-k}e^{-\rho k/2},\quad |\phi_{\ell}(0)| \geqslant \frac{1}{\sqrt{2}}B^{-k}e^{-\rho k/2}.
\end{equation}
Substituting \eqref{phi0} into \eqref{Est} proves the result.
\end{proof}

\section{Lower bound for half-line operators}
Let $Q$ denote the orthogonal projection onto $\ell^{2}([1,\infty)\cap\mathbb{Z})$, that is,
\begin{equation}\label{Q}
    [Q\psi](n)=\begin{cases}\psi(n),&n\geqslant1,\\0,&n\leqslant0.\end{cases}
\end{equation}

\begin{proposition}\label{half}
    Let $\rho>0$ and $k\in\mathbb{N}$. Let 
    \begin{equation*}
        T_{k}=B^{2k} e^{2\rho k}.
    \end{equation*}
    Then, for all sufficiently large $k$,
\begin{equation*}
    \frac{2}{T_{k}}\int_{0}^{\infty}e^{-2t/T_{k}} \|Qe^{-itH_{0}}\delta_{-k}\|^{2}\,\mathrm{d}t\geqslant\frac{1}{4}.
\end{equation*}
\end{proposition}

\begin{proof}
For every $k\in\mathbb{N}$,  let $\mathbb{P}_{k}$ denote the spectral projection onto $\operatorname{span}\{\phi_{s}:s\in S_{k,\rho}\}$, and decompose
\begin{equation}\label{fg}
    \delta_{-k}= \mathbb{P}_{k}\delta_{-k} + (\mathrm{I}-\mathbb{P}_{k})\delta_{-k}\eqqcolon f_{k}+g_{k}.
\end{equation}

We first estimate the quantum dynamics with initial state $f_k$. Denote
\begin{equation*}
    X_{f_{k}}(T)=\frac{2}{T}\int_{0}^{\infty}e^{-2t/T} \|Qe^{-itH_{0}} f_{k}\|^{2}\,\mathrm{d}t.
\end{equation*}
Since $e^{-itH_{0}}\phi_{s}=e^{-itE_{s}}\phi_{s}$ and $f_{k}=\mathbb{P}_{k}\delta_{-k}$, one has
\begin{equation*}
    Qe^{-itH_{0}} f_{k}= \sum_{s\in S_{k,\rho}}e^{-it E_{s}} \phi_{s}(-k) Q\phi_{s},
\end{equation*}
We have
\begin{equation*}
    \begin{split}
        \|Qe^{-itH_{0}} f_{k}\|^{2}&= \sum_{s,\ell\in S_{k,\rho}}e^{-it (E_{s}-E_{\ell})} \phi_{s}(-k)\overline{\phi_{\ell}(-k)} \langle Q\phi_{s},Q\phi_{\ell}\rangle\\
        &=\sum_{s,\ell\in S_{k,\rho}}e^{-it (E_{s}-E_{\ell})} \phi_{s}(-k)\overline{\phi_{\ell}(-k)} \langle \phi_{\ell},Q\phi_{s}\rangle.
    \end{split}
\end{equation*}
Hence
\begin{equation}\label{Xf}
    \begin{split}
        X_{f_{k}}(T)
        &=\sum_{s,\ell\in S_{k,\rho}}\phi_{s}(-k)\overline{\phi_{\ell}(-k)} \langle \phi_{\ell},Q\phi_{s}\rangle \frac{2}{T}\int_{0}^{\infty}e^{-2t/T}e^{-it (E_{s}-E_{\ell})}  \,\mathrm{d}t\\
        &=\sum_{s,\ell\in S_{k,\rho}}\phi_{s}(-k)\overline{\phi_{\ell}(-k)} \langle \phi_{\ell},Q\phi_{s}\rangle \frac{2/T}{2/T+i(E_{s}-E_{\ell})}\\
        &=\sum_{s\in S_{k,\rho}} |\phi_{s}(-k)|^{2} \langle \phi_{s},Q\phi_{s}\rangle\\
        &\quad+ \sum_{s,\ell\in S_{k,\rho},s\neq \ell}\phi_{s}(-k)\overline{\phi_{\ell}(-k)} \langle \phi_{\ell},Q\phi_{s}\rangle \frac{2/T}{2/T+i(E_{s}-E_{\ell})}.
    \end{split}
\end{equation}
For the diagonal terms $s=\ell$, \hyperref[eigenbasis]{Lemma~\ref*{eigenbasis}}  shows that $\langle \phi_{s},Q\phi_{s}\rangle=\frac{1}{2}$, then
\begin{equation*}
    \begin{split}
        \sum_{s\in S_{k,\rho}} |\phi_{s}(-k)|^{2} \langle \phi_{s},Q\phi_{s}\rangle &= \frac{1}{2}\sum_{s\in S_{k,\rho}} |\phi_{s}(-k)|^{2}\\
        &=\frac{1}{2}\bigg(1-\sum_{s\notin S_{k,\rho}}|\phi_{s}(-k)|^{2}\bigg).
    \end{split}
\end{equation*}
It follows from \hyperref[restS]{Lemma~\ref*{restS}} that for sufficiently large $k$,
\begin{equation}\label{s=t}
    \sum_{s\in S_{k,\rho}} |\phi_{s}(-k)|^{2} \langle \phi_{s},Q\phi_{s}\rangle\geqslant \frac{1}{2}-C_{2}e^{-c_{2} k}.
\end{equation}

For the off-diagonal terms where $s\neq \ell$, if $\sigma_{s}=\sigma_{\ell}$, then $\langle \phi_{\ell},Q\phi_{s}\rangle=0$ due to the reflective symmetry and orthogonality. If $\sigma_{s}\neq\sigma_{\ell}$, the eigenvalue separation established in \hyperref[gap]{Lemma~\ref*{gap}} implies that
\begin{equation}\label{gapst}
    \bigg|\frac{2/T}{2/T+i(E_{s}-E_{\ell})}\bigg|\leqslant \frac{2}{T|E_{s}-E_{\ell}|} \leqslant \frac{B^{2k}e^{\rho k}}{T}.
\end{equation}
Moreover, by $|\langle \phi_{\ell},Q\phi_{s}\rangle|\leqslant 1$,
\begin{equation*}
    \begin{split}
        \sum_{\substack{s,\ell\in S_{k,\rho}\\s\neq \ell}}|\phi_{s}(-k)\overline{\phi_{\ell}(-k)} \langle \phi_{\ell},Q\phi_{s}\rangle|&\leqslant \sum_{s,\ell\in S_{k,\rho}}|\phi_{s}(-k)\overline{\phi_{\ell}(-k)}| \\
        &=\bigg(\sum_{s\in S_{k,\rho}}|\phi_{s}(-k)|\bigg)^{2}.
    \end{split}
\end{equation*}
By H\"older inequality, 
\begin{equation*}
    \bigg(\sum_{s\in S_{k,\rho}}|\phi_{s}(-k)|\bigg)^{2}\leqslant \# S_{k,\rho} \sum_{s\in S_{k,\rho}}|\phi_{s}(-k)|^{2}.
\end{equation*}
Thus, by \hyperref[restS]{Lemma~\ref*{restS}} and the completeness of the eigenbasis, we obtain
\begin{equation}\label{stk}
    \sum_{\substack{s,\ell\in S_{k,\rho}\\s\neq \ell}}|\phi_{s}(-k)\overline{\phi_{\ell}(-k)} \langle \phi_{\ell},Q\phi_{s}\rangle|\leqslant \# S_{k,\rho} 
        \leqslant 10D'k.
\end{equation}
By \eqref{Xf}, \eqref{s=t}, \eqref{gapst} and \eqref{stk}, we have for sufficiently large $k$,
\begin{equation*}
    X_{f_{k}}(T)\geqslant \frac{1}{2}-C_{2}e^{-c_{2} k}-10D' \frac{kB^{2k}e^{\rho k}}{T}.
\end{equation*}
Since $T_{k}=B^{2k}e^{2\rho k}$, we have for sufficiently large $k$,
\begin{equation}\label{XfT}
    X_{f_{k}}(T_{k})\geqslant \frac{1}{2}-C_{2}e^{-c_{2} k}-10D'ke^{-\rho k}\geqslant \frac{1}{3}.
\end{equation}

Now we estimate $\|Qe^{-itH_{0}}\delta_{-k}\|$. By $|\|a+b\|^{2}-\|a\|^{2}|\leqslant 2\|a\|\, \|b\|+\|b\|^{2}$ and \eqref{fg}, we have
\begin{equation*}
    |\|Qe^{-itH_{0}}\delta_{-k}\|^{2}-\|Qe^{-itH_{0}}f_{k}\|^{2} |\leqslant 2\|g_{k}\|+\|g_{k}\|^{2}.
\end{equation*}
Since $g_{k}=(\mathrm{I}-\mathbb{P}_{k})\delta_{-k}$, by \hyperref[restS]{Lemma~\ref*{restS}}, we have for sufficiently large $k$,
\begin{equation*}
    \|g_{k}\|^{2}=\sum_{s\notin S_{k,\rho}}|\phi_{s}(-k)|^{2}\leqslant C_{2}e^{-c_{2} k}.
\end{equation*}
Thus for sufficiently large $k$,
\begin{equation}\label{error}
    \|Qe^{-itH_{0}}\delta_{-k}\|^{2}\geqslant \|Qe^{-itH_{0}}f_{k}\|^{2}-3\sqrt{C_{2}}e^{-c_{2} k/2}.
\end{equation}
By \eqref{XfT} and \eqref{error}, for sufficiently large $k$,
\begin{equation*}
    \frac{2}{T_{k}}\int_{0}^{\infty}e^{-2t/T_{k}} \|Qe^{-itH_{0}}\delta_{-k}\|^{2}\,\mathrm{d}t \geqslant X_{f_{k}}(T_{k})-3\sqrt{C_{2}}e^{-c_{2} k/2} \geqslant \frac{1}{4}.
\end{equation*}
\end{proof}

Since we care about the quantum dynamics with initial state $\delta_{0}$, we have the following corollary for the shifted Schr\"odinger 
operator $H_{k}\coloneqq H_{v,\alpha,\theta_{k}}$
where $\theta_{k}\coloneqq \theta_{0}-k\alpha$. 

\begin{corollary}\label{Qk}
    Let $T_{k}=B^{2k} e^{2\rho k}$. For all sufficiently large $k\in\mathbb{N}$,
    \begin{equation*}
        \frac{2}{T_{k}}\int_{0}^{\infty}e^{-2t/T_{k}} \sum_{n\geqslant k+1}|\langle \delta_{n},e^{-itH_{k}}\delta_{0}\rangle |^{2}\,\mathrm{d}t\geqslant\frac{1}{4}.
    \end{equation*}
\end{corollary}
\begin{proof}
    Define the translation operator $U_{k}$ on $\ell^{2}(\mathbb{Z})$ by
    \begin{equation*}
        [U_{k}\psi](n)=\psi(n-k).
    \end{equation*}
    It is obvious that $U_{k}$ is unitary.
    A direct calculation gives $U_{k}H_{0}U_{k}^{-1}=H_{k}$ and thus $e^{-it H_{k}}=U_{k}e^{-it H_{0}}U_{k}^{-1}$.
    
    Let $Q_{k}$ be the orthogonal projection onto $\ell^{2}([k+1,\infty)\cap \mathbb{Z})$. By \eqref{Q}, we have $Q_{k}=U_{k}QU_{k}^{-1}$. Thus
    \begin{equation*}
        \begin{split}
            \sum_{n\geqslant k+1}|\langle \delta_{n},e^{-itH_{k}}\delta_{0}\rangle |^{2} &=\|Q_{k}e^{-itH_{k}} \delta_{0}\|^{2}\\
            &=\|U_{k}QU_{k}^{-1} U_{k}e^{-itH_{0}}U_{k}^{-1}\delta_{0}\|^{2}\\
            &=\|U_{k}Qe^{-itH_{0}}\delta_{-k}\|^{2}\\
            &=\|Qe^{-itH_{0}}\delta_{-k}\|^{2}.
        \end{split}
    \end{equation*}
    Therefore, the result follows from \hyperref[half]{Proposition~\ref*{half}}.
\end{proof}

\section{Proof of the main theorems}

\subsection{Construction of the phase}

Since $v$ is continuous and $\mathbb{T}$ is compact, we have $v$ is uniformly continuous on $\mathbb{T}$. Define the modulus of the continuity of $v$ by
\begin{equation}\label{Lip}
    \omega_{v}(\delta)=\sup\{|v(x)-v(y)|:\|x-y\|_{\mathbb{T}}\leqslant \delta\},\qquad \delta\geqslant0.
\end{equation}

Recall that for every $k\in\mathbb{N}$,
\begin{equation*}
    T_{k}=B^{2k} e^{2\rho k},\qquad\theta_{k}=\theta_{0}-k\alpha.
\end{equation*}

\begin{lemma}\label{phase}
    Let $\rho>0$. There exist $\theta\in\mathbb{T}$ and sequences $\varepsilon_{j}\to 0$ and $k_{j}\to\infty$ such that
    \begin{equation}\label{j1}
        \|\theta-\theta_{k_{j}}\|_{\mathbb{T}}\leqslant \varepsilon_{j},\qquad j\geqslant 1,
    \end{equation}
    \begin{equation}\label{j2}
        \omega_{v}(\varepsilon_{j}) \leqslant T_{k_{j}}^{-1} e^{-(j+4)k_{j}},\qquad j\geqslant 1
    \end{equation}
    \begin{equation}\label{j3}
        \varepsilon_{j}\leqslant e^{-j(2k_{j}+1)}\qquad j\geqslant 1.
    \end{equation}
\end{lemma}

\begin{proof}
    We let $\varepsilon_{0}=1$ and choose $k_{1}$ sufficiently large. This finishes the first step construction.

    Assume that we already have $\varepsilon_{j-1}$ and $k_{j}$. We construct iteratively $\varepsilon_{j}$ and $k_{j+1}$. We first choose $\varepsilon_{j}>0$ such that
    \begin{equation*}
        \omega_{v}(\varepsilon_{j})\leqslant T_{k_{j}}^{-1}e^{-(j+4)k_{j}}\quad \text{and}\quad  \varepsilon_{j}\leqslant \min \bigg\{e^{-j(2k_{j}+1)},\frac{\varepsilon_{j-1}}{2}\bigg\}.
    \end{equation*}
    Next we define
    \begin{equation*}
        k_{j+1}=k_{j}+q_{j},
    \end{equation*}
    where $q_{j}\in\mathbb{N}$ is positive such that
    \begin{equation}\label{qj}
        \|q_{j}\alpha\|_{\mathbb{T}}\leqslant \frac{\varepsilon_{j}}{2}.
    \end{equation}
This finishes the construction and proves \eqref{j2} and \eqref{j3}.

    It is obvious that by \eqref{qj}, for every $j\geqslant 1$,
    \begin{equation}\label{cauchy}
    \|\theta_{k_{j+1}}-\theta_{k_{j}}\|_{\mathbb{T}}\leqslant \|q_{j}\alpha\|_{\mathbb{T}}\leqslant \frac{\varepsilon_{j}}{2}.
\end{equation}
The sequence $\{\theta_{k_{j}}\}$ is Cauchy sequence. Let
\begin{equation*}
    \theta=\lim_{j\to\infty} \theta_{k_{j}}.
\end{equation*}
by \eqref{cauchy}, we have
\begin{equation*}
    \|\theta-\theta_{k_{j}}\|_{\mathbb{T}}\leqslant \sum_{i= j}^{\infty} \|q_{i}\alpha\|_{\mathbb{T}}\leqslant \varepsilon_{j}.
\end{equation*}
This proves \eqref{j1}.
\end{proof}

\begin{remark}
    The phase $\theta$ constructed in \hyperref[phase]{Lemma~\ref*{phase}} satisfies 
    \begin{equation*}
        \delta(\alpha,\theta)\coloneqq \limsup_{|k|\to\infty} \frac{-\ln \|2\theta+k\alpha\|_{\mathbb{T}}}{|k|}=\infty.
    \end{equation*}
    Indeed, combining $\|2\theta_{k_{j}}+(2k_{j}+1)\alpha\|_{\mathbb{T}}=0$ with \eqref{j1} implies that
    \begin{equation*}
        \|2\theta+(2k_{j}+1)\alpha\|_{\mathbb{T}}\leqslant 2\varepsilon_{j}\leqslant 2e^{-j(2k_{j}+1)}.
    \end{equation*}
    Thus, 
    \begin{equation*}
        \delta(\alpha,\theta)
            \geqslant \lim_{j\to \infty} \frac{-\ln \|2\theta+(2k_{j}+1)\alpha\|_{\mathbb{T}}}{2k_{j}+1}=\infty.
    \end{equation*}
\end{remark}

\subsection{Criterion for lower bound}

\begin{proof}[Proof of {\hyperref[criteriontheorem]{Theorem~\ref*{criteriontheorem}}}]
We only need to consider one fixed phase $\theta\in\mathbb{T}$ since the shift of this phase produce a dense subset $\Theta=\cup_{n\in\mathbb{Z}}\{\theta+n\alpha\}$.

Let $\rho>0$ be arbitrarily fixed. Let $B$ be given by \eqref{B}, and let $k_{j}$, $\varepsilon_{j}$, and $\theta$ be given by \hyperref[phase]{Lemma~\ref*{phase}}. We define
\begin{equation*}
    T_{j}=T_{k_{j}}=B^{2k_{j}}e^{2\rho k_{j}}.
\end{equation*}

We first compare the dynamics of $H_{\theta}\coloneqq H_{v,\alpha,\theta}$ with that of the exactly reflected operator $H_{\theta_{k_{j}}}\coloneqq H_{v,\alpha,\theta_{k_{j}}}$. By \eqref{Lip}, we have
\begin{equation}\label{Lipepsilon}
    \|H_{\theta}-H_{\theta_{k_{j}}}\|\leqslant \omega_{v}(\|\theta-\theta_{k_{j}}\|_{\mathbb{T}})\leqslant \omega_{v}(\varepsilon_{j}).
\end{equation}
By Duhamel's formula, for any $t>0$,
\begin{equation}\label{duhamel}
    \|e^{-itH_{\theta}}-e^{-itH_{\theta_{k_{j}}}}\|\leqslant t\|H_{\theta}-H_{\theta_{k_{j}}}\|.
\end{equation}
Let $Q_{j}$ denote the orthogonal projection onto $\ell^{2}([k_{j}+1,\infty)\cap\mathbb{Z})$. Since both time evolutions preserve the $\ell^{2}$ norm, combining \eqref{Lipepsilon} with \eqref{duhamel} yields
\begin{equation}\label{Qerror}
    \begin{split}
        \big|\|Q_{j}e^{-itH_{\theta}}\delta_{0}\|^{2}-\|Q_{j}e^{-itH_{\theta_{k_{j}}}}\delta_{0}\|^{2}\big|&\leqslant 2t \|H_{\theta}-H_{\theta_{k_{j}}}\|\\
        &\leqslant 2t\omega_{v}(\varepsilon_{j}).
    \end{split}
\end{equation}

Using the identity
\begin{equation*}
    \frac{2}{T} \int_{0}^{\infty} te^{-2t/T}\,\mathrm{d}t=\frac{T}{2},
\end{equation*}
we deduce from \eqref{Qerror} that
\begin{equation*}
    \begin{split}
        \bigg|\frac{2}{T_{j}}\int_{0}^{\infty}e^{-2t/T_{j}}\|Q_{j}e^{-itH_{\theta}}\delta_{0}\|^{2}\,\mathrm{d}t- \frac{2}{T_{j}}\int_{0}^{\infty}e^{-2t/T_{j}}&\|Q_{j}e^{-itH_{\theta_{k_{j}}}}\delta_{0}\|^{2}\,\mathrm{d}t\bigg|\\
        &\leqslant  T_{j}\omega_{v}(\varepsilon_{j}).
    \end{split}
\end{equation*}

Applying \hyperref[Qk]{Corollary~\ref*{Qk}} together with \hyperref[phase]{Lemma~\ref*{phase}}, we find that for sufficiently large $j$,
\begin{equation}\label{abeltaillower}
    \begin{split}
        &\quad \frac{2}{T_{j}}\int_{0}^{\infty}e^{-2t/T_{j}} \sum_{n\geqslant k_{j}+1}|\langle \delta_{n},e^{-itH_{\theta}}\delta_{0}\rangle |^{2}\,\mathrm{d}t
        \\
        &\geqslant \frac{2}{T_{j}}\int_{0}^{\infty}e^{-2t/T_{j}} \sum_{n\geqslant k_{j}+1}|\langle \delta_{n},e^{-itH_{\theta_{k_{j}}}}\delta_{0}\rangle |^{2}\,\mathrm{d}t -T_{j}\omega_{v}(\varepsilon_{j})\\
        &\geqslant \frac{1}{4}- e^{-(j+4)k_{j}}\\
        &>\frac{1}{5}.
    \end{split}
\end{equation}

Finally, by \eqref{avmoment},
\begin{equation*}
    \begin{split}
        \langle |\widetilde{X}_{\theta}|^{p}\rangle(T_{j})&\geqslant(k_{j}+1)^{p}\frac{2}{T_{j}}\int_{0}^{\infty}e^{-2t/T_{j}}\sum_{n\geqslant k_{j}+1}|\langle\delta_{n},e^{-itH_{\theta}}\delta_{0}\rangle|^{2}\,\mathrm{d}t\\
        &\geqslant\frac{1}{5}(k_{j}+1)^{p}.
    \end{split}
\end{equation*}
Since $\ln T_{j}=C_{0}k_{j}$, where $C_{0}=2\ln B+2\rho$, we obtain
\begin{equation*}
    \langle |\widetilde{X}_{\theta}|^{p}\rangle(T_{j})\geqslant\frac{1}{5C_{0}^{p}}(\ln T_{j})^{p}.
\end{equation*}
This proves the lower bound for Abel average $\langle |\widetilde{X}_{\theta}|^{p}\rangle$.

Now we prove the lower bound for $\langle |X_{\theta}|^{p}\rangle$. By \eqref{abeltaillower}, we have
\begin{equation*}
    \frac{2}{T_{j}}\int_{0}^{\infty}e^{-2t/T_{j}} \sum_{n\geqslant k_{j}+1}|\langle \delta_{n},e^{-itH_{\theta}}\delta_{0}\rangle |^{2}\,\mathrm{d}t\geqslant \frac{1}{5}.
\end{equation*}
A direct calculation shows that
\begin{equation*}
    \frac{2}{T_{j}}\bigg(\int_{0}^{T_{j}/100}+\int_{100T_{j}}^{\infty}\bigg)e^{-2t/T_{j}} \,\mathrm{d}t < \frac{1}{10}.
\end{equation*}
Consequently,
\begin{equation}\label{midpart}
    \frac{2}{T_{j}}\int_{T_{j}/100}^{100T_{j}}e^{-2t/T_{j}} \sum_{n\geqslant k_{j}+1}|\langle \delta_{n},e^{-itH_{\theta}}\delta_{0}\rangle |^{2}\,\mathrm{d}t\geqslant \frac{1}{10}.
\end{equation}

We claim that there exists $t_{j}\in [T_{j}/100, 100 T_{j}]$ such that
\begin{equation*}
    \sum_{n\geqslant k_{j}+1}|\langle \delta_{n},e^{-it_{j}H_{\theta}}\delta_{0}\rangle |^{2}\geqslant \frac{1}{10}.
\end{equation*}
Indeed, if this were not the case, then
\begin{equation*}
    \begin{split}
        \frac{2}{T_{j}}\int_{T_{j}/100}^{100T_{j}}e^{-2t/T_{j}} \sum_{n\geqslant k_{j}+1}|\langle \delta_{n},e^{-itH_{\theta}}\delta_{0}\rangle |^{2}\,\mathrm{d}t
        &\leqslant \frac{1}{10} \cdot \frac{2}{T_{j}}\int_{T_{j}/100}^{100T_{j}}e^{-2t/T_{j}} \,\mathrm{d}t\\
        &< \frac{1}{10},
    \end{split}
\end{equation*}
which contradicts \eqref{midpart}.

Thus, along the sequence $t_{j}\to\infty$, we obtain
\begin{equation*}
    \langle |X_{\theta}|^{p}\rangle(t_{j})\geqslant (k_{j}+1)^{p}\sum_{n\geqslant k_{j}+1}|\langle \delta_{n},e^{-it_{j}H_{\theta}}\delta_{0}\rangle |^{2} \geqslant \frac{1}{10(2C_{0})^{p}} (\ln t_{j})^{p}.
\end{equation*}
This completes the proof.
\end{proof}

\subsection{Logarithmic transport exponent}

\begin{proof}[Proof of {\hyperref[sharpthm]{Theorem~\ref*{sharpthm}}}]
    We only give the proof for $\langle|X_{\theta}|^{p}\rangle$ since $\langle|\widetilde{X}_{\theta}|^{p}\rangle$ is similar. 

    By \hyperref[criteriontheorem]{Theorem~\ref*{criteriontheorem}}, we have
    \begin{equation*}
        \limsup_{t\to\infty}\frac{\ln \sup_{\theta}\langle |X_{\theta}|^{p}\rangle(t)}{p\ln \ln t} \geqslant \limsup_{j\to\infty}\frac{\ln \langle |X_{\theta}|^{p}\rangle(t_{j})}{p\ln \ln t_{j}} \geqslant 1.
    \end{equation*}
    On the other hand, for every $p>0$, \eqref{upper} gives
    \begin{equation*}
        \limsup_{t\to\infty}\frac{\ln \sup_{\theta}\langle |X_{\theta}|^{p}\rangle(t)}{p\ln \ln t} \leqslant \frac{p+\varepsilon}{p}.
    \end{equation*}
    Letting $\varepsilon\to 0$ finishes the proof.
\end{proof}

\subsection{Application to the almost Mathieu operator}

Let $H_{\lambda,\alpha,\theta}$ be the almost Mathieu operator given by \eqref{amo}. We first recall the following two results for completely resonant phases. 

For Diophantine frequencies, the following Anderson localization result was established in \cite{MR2121278}.

\begin{theorem}\label{AL}
    Assume that $\alpha\in\mathrm{DC}$, $2\theta\in\alpha\mathbb{Z}+\mathbb{Z}$, and $\lambda>1$. Then  $H_{\lambda,\alpha,\theta}$ exhibits Anderson localization.
\end{theorem}

The following result was essentially proved in \cite{MR4216568,MR4756946}. The version formulated here is taken from \cite{jitomirskaya2024sharp}*{Corollary 3.3}.

\begin{theorem}\label{JLM}
    For every $\varepsilon>0$, there exists $K=K(\alpha,\lambda,\varepsilon)>0$ such that the following holds. Let $\phi_{s}$ be a normalized eigenfunction of $H_{\lambda,\alpha,\theta}$ and let $m_{s}$ be global maximum point of $|\phi_{s}|$. 
    Denote
    \begin{equation*}
        U^{\phi_{s}}(n)=\begin{pmatrix}
            \phi_{s}(n)\\
            \phi_{s}(n-1)
        \end{pmatrix}.
    \end{equation*}
    Let $k\geqslant K$. Assume $k_{0}\in [-30 k, 30 k]$ satisfies
    \begin{equation*}
        |\sin \pi (2\theta+k_{0}\alpha)|=\min_{|x|\leqslant 30 k} |\sin \pi (2\theta+x\alpha)|.
    \end{equation*}
    If $|m_{s}|\leqslant k$ and $|n|\leqslant 3k$,
    \begin{equation*}
        \begin{split}
            \|U^{\phi_{s}}(n)\|\leqslant e^{\varepsilon k} \max \big\{&\|U^{\phi_{s}}(m_{s})\| e^{-(\ln \lambda-\varepsilon)|n-m_{s}|},\\
            &\|U^{\phi_{s}}(k_{0}-m_{s})\|e^{-(\ln \lambda-\varepsilon)|n-(k_{0}-m_{s})|}\big\}.
        \end{split}
    \end{equation*}
\end{theorem}

We now finish the proof.
\begin{proof}[Proof of {\hyperref[amothm]{Theorem~\ref*{amothm}}}]
Fix  $\alpha\in\mathrm{DC}$ and  $\lambda>1$. By \cite{MR1933451}, the Lyapunov exponent of the almost Mathieu operator is $\ln \lambda>0$ in its spectrum. Let $\theta_{0}=-\alpha/2$. We write $H_{0}=H_{\lambda,\alpha,\theta_{0}}$ for brevity. By \hyperref[criteriontheorem]{Theorem~\ref*{criteriontheorem}} and \hyperref[sharpthm]{Theorem~\ref*{sharpthm}}, we only need to prove that $H_{0}$ satisfies reflective SULE.

By \hyperref[AL]{Theorem~\ref*{AL}}, the operator $H_{0}$ has Anderson localization, i.e. pure point with exponential decay eigenfunctions. Thus $H_{0}$ has a complete orthonormal basis $\{\phi_{s}\}_{s\in\mathbb{Z}}$ with corresponding eigenvalues $\{E_{s}\}_{s\in\mathbb{Z}}$.

Now we study $\phi_{s}$. Fix any $0<\epsilon<\frac{1}{2}\ln\lambda$. For every $s$, choose a global maximum point $m_{s}\leqslant0$ of $|\phi_{s}|$. Apply \hyperref[JLM]{Theorem~\ref*{JLM}} to the operator $H_{0}$. Note that $k_{0}=1$ since
    \begin{equation*}
        |\sin \pi (2\theta_{0}+\alpha)|=0.
    \end{equation*}
    Let $K=K(\alpha,\lambda,\frac{\epsilon}{2})>0$ be the constant in  \hyperref[JLM]{Theorem~\ref*{JLM}}. 
    Fix $s$ and $n$, we take
    \begin{equation}\label{defk}
        k=\max \bigg\{K, |m_{s}|, \bigg[\frac{|n|}{3}\bigg]+1\bigg\}.
    \end{equation}
    Thus we have $k\geqslant K$, $|m_{s}|\leqslant k$, and $|n|\leqslant 3k$, thus all the conditions in  \hyperref[JLM]{Theorem~\ref*{JLM}} are satisfied. Hence for any $s,n\in\mathbb{Z}$,
    \begin{equation*}
        \begin{split}
            \|U^{\phi_{s}}(n)\|\leqslant e^{\frac{\epsilon}{2} k} \max \{&\|U^{\phi_{s}}(m_{s})\| e^{-(\ln \lambda-\frac{\epsilon}{2})|n-m_{s}|},\\
            &\|U^{\phi_{s}}(1-m_{s})\|e^{-(\ln \lambda-\frac{\epsilon}{2})|n-(1-m_{s})|}\}.
        \end{split}
    \end{equation*}
    By $\|\phi_{s}\|_{\ell^{2}}=1$ and \eqref{defds},
    we have
    \begin{equation}\label{decayphi}
        |\phi_{s}(n)|\leqslant 2e^{\frac{\epsilon}{2} k}e^{-(\ln\lambda -\frac{\epsilon}{2}) d_{s}(n)}.
    \end{equation}
    Next, we estimate $k$. By \eqref{defds}, if $d_{s}(n)=|n-m_{s}|$, then
    \begin{equation*}
        |n|\leqslant |m_{s}|+d_{s}(n).
    \end{equation*}
    If $d_{s}(n)=|n-(1-m_{s})|$, then
    \begin{equation*}
        |n|\leqslant |1-m_{s}|+d_{s}(n).
    \end{equation*}
    Thus for both cases, we have $|n|\leqslant |m_{s}|+d_{s}(n)+1$. By \eqref{defk}, we have
    \begin{equation}\label{estk}
        k\leqslant K+|m_{s}|+1+|n|\leqslant K+2+2 |m_{s}|+d_{s}(n).
    \end{equation}
    Substitute \eqref{estk} into \eqref{decayphi}, one obtain
    \begin{equation*}
        |\phi_{s}(n)|\leqslant 2e^{\frac{\epsilon}{2}(K+2)} e^{\epsilon |m_{s}|}e^{-(\ln \lambda-\epsilon)d_{s}(n)}.
    \end{equation*}
    Setting $\gamma=\ln\lambda-\epsilon$ and $C_{\epsilon}=2e^{\frac{\epsilon}{2}(K+2)}$, we conclude that $H_{0}$ has reflective SULE by \hyperref[reflectivelocalization]{Definition~\ref*{reflectivelocalization}}. 
\end{proof}

\section*{Acknowledgments}
Part of this work was completed while Wencai Liu was visiting the Simons Institute for the Theory of Computing. The authors thank Shiwen Zhang for helpful comments on an earlier version of the manuscript. This research was supported in part by NSF grant DMS-2246031.

\section*{Statements and Declarations}
{\bf Conflict of Interest} 
The authors declare no conflicts of interest.
				
\vspace{0.2in}
{\bf Data Availability}
Data sharing is not applicable to this article as no new data were created or analyzed in this study.

\bibliography{main}
\end{document}